\documentclass[11pt,english]{article}

\usepackage[letterpaper]{geometry}
\usepackage{amsthm}
\usepackage{amsmath}
\usepackage{setspace}
\usepackage{amssymb}
\usepackage[all]{xy}
\makeatletter
\theoremstyle{plain}
\newtheorem{thm}{Theorem}[section]
  \theoremstyle{definition}
  \newtheorem{defn}[thm]{Definition}
  \theoremstyle{plain}
  \newtheorem{prop}[thm]{Proposition}
  \theoremstyle{plain}
  \newtheorem{lem}[thm]{Lemma}
 \theoremstyle{definition}
  \newtheorem{example}[thm]{Example}

\makeatother

\newcommand{\ignore}[1]{}

\usepackage{babel}

\begin{document}

\title{Oblivious Algorithms for the Maximum Directed Cut Problem}

\author{Uriel Feige\thanks{Weizmann Institute of Science, Rehovot, Israel. Email: uriel.feige@weizmann.ac.il. The author holds the Lawrence
G. Horowitz Professorial Chair at the Weizmann Institute. Work
supported in part by The Israel Science Foundation (grant No.
873/08).}~ and Shlomo Jozeph\thanks{Weizmann Institute of Science,
Rehovot, Israel. Email: shlomo.jozeph@weizmann.ac.il. Work
supported in part by The Israel Science Foundation (grant No.
873/08).}}

\maketitle

\thispagestyle{empty}

\begin{abstract}

This paper introduces a special family of randomized algorithms for Max DICUT that
we call oblivious algorithms. Let the bias of a vertex be the ratio between the total
weight of its outgoing edges and the total weight of all its edges. An oblivious algorithm
selects at random in which side of the cut to place a vertex v, with probability that only
depends on the bias of v, independently of other vertices. The reader may observe that
the algorithm that ignores the bias and chooses each side with probability 1/2 has an
approximation ratio of 1/4, whereas no oblivious algorithm can have an approximation
ratio better than 1/2 (with an even directed cycle serving as a negative example). We
attempt to characterize the best approximation ratio achievable by oblivious algorithms,
and present results that are nearly tight. The paper also discusses natural extensions of
the notion of oblivious algorithms, and extensions to the more general problem of Max
2-AND.

\end{abstract}

\newpage

\section{Introduction}

\setcounter{page}{1}

Given a directed graph $G=\left(V,E,w\right)$ where
$w:E\to\mathbb{R}^{+}$ is a weight function, and a set $S\subseteq
V$, the \emph{weight} of the directed cut defined by $S$ is the
sum of $w\left(\left(u,v\right)\right)$ with $u\in S$, $v\notin
S$. The goal of the \emph{maximum directed cut} (Max DICUT)
problem is finding a set such that the weight of the respective
directed cut is as large as possible. The \emph{relative weight}
of a cut is the weight of the cut divided by the sum of the
weights of all edges.

The \emph{maximum cut} (Max CUT) problem is a similar problem; $G$
is undirected and the cut contains those edges with exactly one
endpoint in $S$. Max CUT can be seen as a restriction of Max DICUT
with two additional conditions: $\left(u,v\right)\in E$ iff
$\left(v,u\right)\in E$ and every two antisymmetric edges have the
same weight. Except in section \ref{Related-Work}, the term
{}``cut'' will mean directed cut, all graphs will be directed
graphs, and, unless stated otherwise, all graphs will be weighted
graphs.

Given a set of boolean variables $V$, a 2-AND formula is a set of
clauses $C$, where each clause is a conjunction of two different
literals (where a literal is a variable with either positive or
negative polarity). Given a nonnegative weight function
$w:C\to\mathbb{R}^{+}$ over the clauses, the \emph{weight} of an
assignment for the variables is the sum of weights of satisfied
clauses. \emph{Max 2-AND} is the problem of finding an assignment
with maximum weight in a 2-AND formula.

Max DICUT is a special case of the Max 2-AND: Given a graph
$G=\left(V,E,w\right)$, the set of variables will be $V$ and each
edge will define a constraint that is true iff the first vertex is
selected (the corresponding variable is true) and the second
vertex is not selected (the corresponding variable is false).

\begin{defn}An edge $\left(u,v\right)$ is an \emph{inedge} for $v$
and an \emph{outedge} for $u$. The \emph{outweight} 
of a vertex is the sum of the weight of its outedges and the \emph{inweight}
of a vertex is the sum of the weight of its inedges.\end{defn}
\begin{defn}
The \emph{bias} of a vertex is its outweight divided by the sum of
its outweight and its inweight. The bias of a variable is the
weight of the clauses in which it appears positively divided by
the total weight of the clauses it appears in.

An \emph{oblivious algorithm} for Max DICUT selects each vertex to
be in $S$ with some probability that depends only on its bias and
the selection of each vertex is independent of whether other
vertices are selected. Similarly, an oblivious algorithm for Max
2-AND selects each variable to be true with some probability that
depends only on its bias. The \emph{selection function} of an
oblivious algorithm is the function that maps a vertex's (or
variable's) bias to the probability it is selected. All vertices
must use the same selection function.
\end{defn}
Note that a selection function uniquely determines an oblivious
algorithm, so there will be no distinction between them in the
text.

It will be assumed that the probabilities of selecting a vertex
(or a variable) are antisymmetric. That is, if $f$ is the
selection function of an oblivious algorithm then for all biases
$x\in\left[0,1\right]$, $f\left(x\right)+f\left(1-x\right)=1$, or
equivalently, $f\left(1-x\right)=1-f\left(x\right)$.
This assumption seems natural, since with it, oblivious algorithms are invariant to reversing the direction of all edges of the graph. The assumption will be used in Section
\ref{sec:Equivalence-of-Expectation} and to get a better upper bund on the approximation ratio of oblivious algorithms.

The \emph{approximation ratio} of an oblivious algorithm on a
\emph{specific} graph is the expected weight of the cut produced
by the algorithm divided by the weight of the optimal cut. The
approximation ratio of an oblivious algorithm is the infimum of
the approximation ratios on all graphs. The approximation ratio of
an oblivious algorithm for max 2-AND is defined similarly. The
approximation ratio of an oblivious algorithm will be used as a
measure for the quality of the algorithm.

An oblivious algorithm with positive approximation ratio  must be random. Otherwise, in a graph where
all neighborhoods look the same, such as a cycle, all vertices will
belong to $S$ or no vertices will belong to $S$, so the weight of
the cut will be 0.

We are primarily interested in oblivious algorithms, but we will also
discuss two ways of using finite sets of oblivious algorithms. One
is a \emph{mixed} oblivious algorithm, that is, choosing an algorithm
to use from the set according to some (fixed) probability distribution.
The other is \emph{max} of oblivious algorithms, that is, using all
the algorithms in the set to generate cuts and outputting the cut
with the maximal weight.

The approximation ratio for a mixed algorithm is its expected
approximation ratio (where expectation is taken both over the
choice of oblivious algorithm from the set, and over the
randomness of the chosen oblivious algorithm).

There are two natural ways to define the approximation ratio of a
max algorithm: either using \emph{maxexp} -- the maximum (over all
oblivious algorithms in the set) of the expected weight of the
cut, or using \emph{expmax} -- the expectation of the weight of
the maximum cut.
Observe that maxexp cannot be better than expmax, but expmax can
be better than maxexp. For example, assume the set is a multiset
containing a single algorithm multiple times. Then, expmax is
equal to the approximation ratio of the algorithm, but maxexp may
be better. However, it will be shown that the worst case
approximation ratio when using expmax is the same as maxexp.

\subsection{\label{Related-Work}Related work}


Our notion of oblivious algorithms can be viewed as a restricted special
case of the notion of local algorithms used in distributed
computing, which have been studied due to their simplicity, running time,
and other useful characteristics  \cite{S10}.

The uniformly random algorithm selects each vertex (or sets each
variable to true) independently with probability $\frac{1}{2}$. It
gives a $\frac{1}{4}$ approximation to Max 2-AND and a
$\frac{1}{2}$ approximation to Max CUT. There are algorithms that
use semidefinite programming to achieve about 0.874 approximation
to Max 2-AND \cite{LLZ02} and about 0.878 approximation to Max CUT
\cite{GW95}. Assuming the Unique Games Conjecture, these
algorithms are optimal for Max CUT \cite{KKMO04,MOO05}, and nearly
optimal for Max 2-AND (which under this assumption is hard to
approximate within 0.87435 \cite{A07}). Earlier NP-Hardness
results are $\frac{11}{12}$ for Max 2-AND and $\frac{16}{17}$ for
Max CUT~\cite{H01}.

Trevisan \cite{T98} shows how to get $\frac{1}{2}$ approximation
to Max 2-AND using randomized rounding of a linear program.
Halperin and Zwick \cite{HZ01} show simple algorithms that achieve
$\frac{2}{5}$ and $\frac{9}{20}$ approximation ratios, and a
combinatorial algorithm that finds a solution to the previous
linear program.

Bar-Noy and Lampis \cite{BnL09} present an online version of Max
DICUT for acyclic graphs. Vertices are revealed in some order
(respecting the order defined by the graph), along with their
inweight, outweight, and edges to previously revealed vertices,
and based on this information alone they are placed in either side
of the cut. They show that an algorithm achieving an approximation
ratio of $\frac{2}{3\sqrt{3}}$ is optimal against an adaptive
adversary. They also show that derandomizing the uniformly random
algorithm gives an approximation ratio of $\frac{1}{3}$. Oblivious
algorithms can be used in online settings, and in fact, they do
not require the graph to be acyclic and do not require edges to
previously revealed vertices to be given. The reason why the
approximation ratios in the current manuscript are better than
$\frac{2}{3\sqrt{3}}$ is that our approximation ratio is computed
against an oblivious adversary.

Alimonti shows a local search algorithm that achieves an
approximation ratio of $\frac{1}{4}$ for Max 2-AND \cite{A96}, and
uses non-oblivious local search to achieve a $\frac{2}{5}$
approximation \cite{A97}.

Alon et al. \cite{ABGLS07} show that the minimal relative weight
of a maximum directed cut in acyclic unweighted graphs is
$\frac{1}{4}+o\left(1\right)$. Lehel, Maffray and Preissmann
\cite{LMP09} study the minimal weight of a maximum directed cut
(in unweighted graphs) where the indegree or outdegree of all
vertices is bounded. They show that the smaller the degree the
larger the maximum cut. If the indegree or outdegree is 1 for all
vertices, the minimal relative weight is $\frac{1}{3}$. If the
graph also has no directed triangles, the minimal relative weight
is $\frac{2}{5}$.

Feige, Mirrokni and Vondrak \cite{FMV07} show an algorithm that
achieves a $\frac{2}{5}$ approximation to any nonnegative
submodular function (and directed cut is a special case of a
submodular function).


\subsection{Our results}
The main results of the paper are theorems \ref{0.483 approximation} and
\ref{upper bound}, that show that there is an oblivious algorithm that achieves
an approximation ratio of 0.483, but no oblivious algorithm can achieve an
approximation ratio of 0.4899. In the process of proving these theorems,
a few other interesting results are shown.

Max DICUT is a special case of Max 2-AND, and hence approximation
algorithms for Max 2-AND apply to Max DICUT as well. The following
theorem shows a converse when oblivious algorithms are concerned.

\begin{thm}
\label{thm:Equality-of-Expectation}Given any antisymmetric selection function
$f$, the approximation ratio of the corresponding oblivious
algorithm for Max 2-AND is the same as that for Max DICUT.
\end{thm}
Hence our results concerning oblivious algorithms for Max DICUT
extend to Max 2-AND. We remark that for general approximation
algorithms, it is not known whether Max 2-AND can be approximated
as well as Max DICUT (see~\cite{A08} for example).

The following rather standard proposition justifies the use of
expected approximation ratio as the measure of quality of a
randomized approximation algorithm. It also holds for {\em mix}
and {\em max} algorithms, and implies that expmax is not
better than maxexp in the worst case.
\begin{prop}
Given a graph $G$ and $\epsilon>0$, there is another graph $G_{\epsilon}$
such that with high probability, for any oblivious algorithm, with high probability, the
weight of the cut produced by running the algorithm on $G_{\epsilon}$
is close to the expected weight of the cut on $G$ up to $\pm\epsilon$,
and the weight of the optimal cut of both graphs is the same.\end{prop}
\begin{proof}
[Proof (Sketch)] This proposition follows immediately from the law
of large numbers and a graph composed of many disjoint copies of the original
graph (with the weight of edges normalized to be the same as for the original
graph).
\end{proof}
When using the same set of algorithms, the {\em max} algorithm is
not worse than any {\em mixed} algorithm. The following theorem
shows that the converse holds for some mixed algorithm.

\begin{thm}
\label{best=00003Dmixed}Given a finite set of algorithms, there is
a mixed algorithm over the set such that the worst case approximation
ratio is as good as that of the max algorithm of the set.
\end{thm}
$f:\left[0,1\right]\to\left[0,1\right]$ is a \emph{step function}
if there are $0=z_{0}<z_{1}<\cdots<z_{n}<z_{n+1}=1$ such that $f$
is constant on $\left(z_{i},z_{i+1}\right)$. We  first show a simple step
function that has an approximation ratio that is better than the trivial oblivious algorithm. Unlike the proof of theorem \ref{0.483 approximation}, which is computer assisted, we show a complete analytic proof of the following result.\begin{thm}
\label{thm:There-Steps}There is step function with three steps
such that the corresponding oblivious algorithm has approximation
ratio $\frac{3}{8}$.
\end{thm}
We will primarily consider step functions because any function can
be approximated using a step function, in the sense that the step
function will have an approximation ratio that is worse by at most
an arbitrarily small constant (that depends on the width of the steps).
In addition, we will show how to compute the approximation ratio of
any step function.\begin{thm}
\label{linear program lower bound} Given a selection function that
is a step function with $m$ steps, the approximation ratio of the
corresponding oblivious algorithm can be computed as the solution
of a linear program with $\mathrm{O}\left(m\right)$ constraints
and $\mathrm{O}\left(m^{2}\right)$ variables.
\end{thm}
Using a linear program to find the approximation ratio of an algorithm
is referred to as \emph{Factor Revealing Linear Programs} and was
used to find the approximation ratio of algorithms for facility location
\cite{JMMSV03}, $k$-set cover \cite{ACK09}, and buffer management
with quality of service \cite{B04}. It was also used to find the
best function to use in an algorithm for matching ads to search results
\cite{MSVV07}.
\begin{thm}
\label{0.483 approximation}There is an oblivious algorithm with a
step selection function that achieves an approximation ratio of at
least 0.483.
\end{thm}
We provide a computer assisted proof of Theorem~\ref{0.483
approximation}, using the linear programming approach of Theorem
\ref{linear program lower bound}.

The family of linear programs can be used to find the best oblivious
algorithm, up to some additive factor.
\begin{thm}
\label{search best}Given $n\in\mathbb{N}$, there is an algorithm
that uses time $\mathrm{poly}\left(n\right)n^{n}$ to find the best
oblivious algorithm up to an additive factor of
$\mathrm{O}\left(\frac{1}{n}\right)$.
\end{thm}
A trivial upper bound on the approximation ratio of every
oblivious algorithm is $\frac{1}{2}$. For a directed even cycle,
the maximum cut has relative weight $\frac{1}{2}$, whereas an
oblivious algorithm can capture at most one quarter of the edges,
in expectation. We improve this upper bound on the approximation
ratio and show that the function from Theorem \ref{0.483
approximation} is very close to being optimal.
\begin{thm}
\label{upper bound}There is a weighted graph for which the approximation
ratio of any oblivious algorithm (with an antisymmetric selection function) is less than 0.4899.
\end{thm}
Since the upper bound is shown by a single graph, the bound holds not only for a single oblivious algorithm, but also for mixed and max
algorithms.

Analyzing the approximation ratios of oblivious algorithms on
weighted and unweighted graphs is practically the same. The proof
of Proposition~\ref{weighted=00003Dunweighted} follows standard
arguments (see \cite{CST96}, for example) and appears in
Section~\ref{app:weighted} in the appendix.

\begin{prop}
\label{weighted=00003Dunweighted} For every oblivious algorithm
the approximation ratio is the same for weighted and unweighted
graphs.
\end{prop}

Theorem \ref{upper bound} uses the fact that selection functions
are antisymmetric. One might think that this is what prohibits us
from reaching an approximation ratio of $\frac{1}{2}$. However,
even selection functions that are not antisymmetric cannot achieve
an approximation ratio of $\frac{1}{2}$, or arbitrarily close to
$\frac{1}{2}$.
\begin{thm}
\label{nonsymmetric}There is a constant $\gamma>0$ such that any
oblivious algorithm, even one not using an antisymmetric selection
function, has an approximation ratio at most $\frac{1}{2}-\gamma$.
\end{thm}

\section{2-And versus Directed Cut\label{sec:Equivalence-of-Expectation}}

In this section we prove Theorem
\ref{thm:Equality-of-Expectation}. The theorem follows from the
next lemma:
\begin{lem}
Given an instance of Max 2-AND, $\varphi$, there is a graph
$G_{\varphi}$, such that the approximation ratio of an oblivious
algorithm on $\varphi$, using a selection function $f$, is not
worse than an oblivious algorithm on $G_{\varphi}$, using the same
selection function.\end{lem}
\begin{proof}
Consider an instance $\varphi=\left(V,C,w\right)$ of Max 2-AND. We
will create a directed graph $G_{\varphi}=\left(V',E,w'\right)$.
$V'=\left\{ x,\bar{x}|x\in V\right\} $, the set of all literals.
For any clause $c\in C$, $c=y\wedge z$ (where $y,z$ are literals)
there are two edges in $E$: one from the vertex $y$ to the vertex
corresponding to the negation of $z$ and another from the vertex
$z$ to the vertex corresponding to the negation of $y$. Each of
these edges has weight $\frac{1}{2}w\left(c\right)$.

Every assignment for $\varphi$ can be transformed to a cut for
$G_{\varphi}$ of the same weight, trivially, by selecting all (and
only) literals (as vertices in the graph $G_{\varphi}$) that are
true in the assignment. Hence the optimal cut weighs at least as
much as the optimal assignment. Note, however, that the converse
does not hold. For example, for the following set of clauses:
$\left\{ x\wedge y,\bar{x}\wedge
y,x\wedge\bar{y},\bar{x}\wedge\bar{y}\right\}$ the weight of the
optimal assignment is~1, whereas the optimal cut in the graph has
weight~2. (Select $x$ and $\bar{x}$, a selection that does not
correspond to an assignment.)

The expected weight of an assignment for $\varphi$ is equal to the
expected weight of a cut in $G_{\varphi}$, when using oblivious
algorithms with the same selection function. Note that the bias of
a vertex is equal to the bias of the corresponding literal (where
the bias of a negation of a variable is one minus the bias of the
variable). Thus, the respective probabilities are equal. Hence,
the probability of  any clause being satisfied is equal to the
probability of each of the two edges generated from the clause
being in the cut. Since the weight of the edges is one half of the
weight of the clause, and due to the linearity of expectation, the
claim follows.
\end{proof}

See Section~\ref{sec:andcutremarks} in the appendix for some
remarks on the above proof.

\section{Mix versus Max\label{sec:Mix-versus-Max}}

In this section we will prove Theorem \ref{best=00003Dmixed}. We
first present an example that illustrates the contents of the
theorem.

The uniformly random algorithm selects every vertex independently
with probability $\frac{1}{2}$.
\begin{prop}
The uniformly random algorithm has an approximation ratio of
$\frac{1}{4}$.\end{prop}
\begin{proof}
An edge is expected to be in the cut with probability
$\frac{1}{4}$ (each vertex of the edge in one side and in the
correct direction) and the weight of the cut is at most all the
edges.
\end{proof}
The greedy algorithm selects a vertex if the outweight is larger than
the inweight (for equal weights the selection can be arbitrary).
\begin{prop}
If the relative weight of the maximal cut is $1-\epsilon$, then the
greedy algorithm produces a cut of relative weight at least $1-2\epsilon$.\end{prop}
\begin{proof}
Consider a maximum cut in the graph of relative weight
$1-\epsilon$. An endpoint of an edge is said to be misplaced by an
algorithm if it is an outedge not placed in $S$ or an inedge that
is placed in $S$. An edge is not in the cut iff at least one of
its endpoints is misplaced. The relative weight of endpoints not
in the optimal
cut is $2\epsilon$.\\
Now, estimate the relative weight of edges not in the cut produced
by the greedy algorithm, using the edges' endpoints. The greedy algorithm
minimizes the weight of the endpoints not in the cut, but every endpoint
may correspond to an edge. Since the estimate is at most $2\epsilon$,
the relative weight of the edges is at most $2\epsilon$.
\end{proof}
Let us consider the max of the uniformly random algorithm and the
greedy algorithm. The approximation ratio is $\frac{2}{5}$: when
the weight of the maximal cut is at most $\frac{5}{8}$ of the
edges, the uniformly random algorithm will give an approximation
ratio of at least $\frac{2}{5}$ and when at most $\frac{3}{8}$ are
not in the cut, the greedy algorithm will give an approximation
ratio of at least $\frac{2}{5}$ ($\frac{1-2\epsilon}{1-\epsilon}$
is a decreasing function).

This approximation ratio is optimal: 
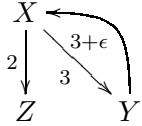
\[
\xymatrix{X\ar[d]_{2}\ar[dr]_{3}\\
Z & Y\ar@(u,r)[ul]^{3+\epsilon}}
\]
Selecting $X$ gives an optimal cut of weight $5$, but the greedy
algorithm will select both $X$ and $Y$, so the cut produced will
have weight 2. The uniformly random algorithm is expected to
produce a cut of weight $2+\frac{\epsilon}{4}$.

Let us now consider a mixed algorithm using the two algorithms. Let
$1-\epsilon$ be the relative weight of the cut. A mixed algorithm
using the greedy algorithm with probability $\gamma$ and the uniformly
random otherwise will give a cut of relative weight $\left(1-\gamma\right)\frac{1}{4}+\gamma\left(1-2\epsilon\right)$.\\
For $\gamma=\frac{1}{5}$, the mixed algorithm gives an
approximation ratio of $\frac{2}{5}$.

The equality of the approximation ratios of the max and mixed
algorithms is not accidental. Define a two player zero sum game:
Player A (for algorithm) has a finite set of pure strategies
corresponding to oblivious algorithms. Player G (for graph) has
pure strategies corresponding to all graphs. When both players use
pure strategies, player A is given a payoff equal to the
approximation ratio of the algorithm (corresponding to the
selected strategy by A) on the graph (corresponding to the
selected strategy by G). A max algorithm (using the maxexp notion
of approximation ratio) is the same as allowing player A to select
a pure strategy after player G has chosen a pure strategy. A mixed
algorithm is the same as allowing player A to use a mixed strategy.
By the minimax theorem, the best mixed strategy gives the same
payoff as having first G choose a mixed strategy (a distribution
over graphs), and then A chooses the best pure strategy against
this distribution.\footnote{The argument is a bit more delicate
because G has infinitely many pure strategies. A form of the
Minimax theorem holds also in this case since the payoffs are
bounded (see for example Theorem 3.1 in \cite{W45}).  For any
$\epsilon>0$ there is a value $\beta$, such that A has a mixed
strategy with payoff at least $\beta-\epsilon$, and G has a mixed
strategy limiting the payoff to be at most $\beta$.}~  Now the key
observation showing equality (up to arbitrary precision) between
mixed and max is that every mixed strategy for G can be
approximated by a pure strategy of G. Player G can choose a single
graph instead of a distribution of graphs: By losing at most
$\epsilon$ of the payoff (for any $\epsilon>0$), it can be assumed
that the distribution over the graphs is rational and finitely
supported. That is, the mixed strategy is
$\left(\frac{p_{1}}{M},\cdots,\frac{p_{n}}{M}\right)$, where
$p_{i},M\in\mathbb{N}$ and $\frac{p_{i}}{M}$ is the probability of
selecting the graph $G_{i}$. Construct $G^*$ from a disjoint union
of $G_{i}$ (for $1\leq i\leq n$), and multiply the weights of the
edges of the copy of $G_{i}$ in $G^*$ by the inverse of the weight
of the optimal cut in $G_{i}$ times $p_{i}$ (so that the weight of
the optimal cut in $G^*$ is 1). On $G^*$, no pure strategy of A
gives an approximation ratio better than $\beta+\epsilon$ (where $\beta$ is the value of the game).
Hence, given a set of oblivious algorithms, a max algorithm is not
better (up to arbitrary precision) than the best mixed algorithm.

Note that a mixed algorithm (over a set of oblivious algorithms)
is not an oblivious algorithm. We do not know if there are mixed
algorithms with worst case approximation ratios better than those
for oblivious algorithms.

\section{\label{ThreeSteps}An oblivious algorithm with $\frac{3}{8}$
approximation ratio}

In this section we will prove Theorem \ref{thm:There-Steps}.

There is another way to {}``mix'' between the greedy and uniform
algorithm. Consider the family of selection functions
$f_{\delta}$, for $0<\delta<\frac{1}{2}$, where \[
f_{\delta}\left(x\right)=\left\{ \begin{array}{cc}
0 & 0<x<\delta\\
\frac{1}{2} & \delta\leq x\leq1-\delta\\
1 & 1-\delta<x<1\end{array}\right.\]

Fix $\delta$ and consider a graph $G$. Divide its vertices into
two sets: $U$ (unbalanced), which will contain all vertices with
bias at most $\delta$ or more than $1-\delta$, and $B$ (balanced),
the rest of the vertices. We will modify (by adding vertices and
removing and adding edges) the graph to have a simpler structure
while making sure that the expected weight of the cut will not
increase (when using $f_{\delta}$), the weight of the optimal cut
will not decrease, and that the biases of all original vertices
will remain the same. Divide $U$ further into $U^{+}$, the set of
vertices of bias more then $1-\delta$, and $U^{-}$, the set of
vertices of bias less then $\delta$. Consider an edge
$\left(u,v\right)$ between $U^{+}$ to $U^{-}$ (in either
direction). Add a new vertex, $w$ and two new edges
$\left(u,w\right)$ and $\left(w,v\right)$, both with the same
weight of the edge $\left(u,v\right)$ and remove the edge
$\left(u,v\right)$.

After the transformation of $G$, there are no edges between
$U^{+}$ and $U^{-}$. By normalizing the weight of the graph, we
may assume that the sum of weights of edges is 1. Let $p^{+}$ be
the weight of endpoints in $U^{+}$, $p^{-}$ the weight of
endpoints in $U^{-}$, so there are $2-p^{-}-p^{+}$ endpoints in
$B$. Let $\gamma^{-}p^{-}$ be the weight of endpoints misplaced by
$f_{\delta}$ in $U^{-}$ and let $\gamma^{+}p^{+}$ be the weight of
endpoints misplaced by $f_{\delta}$ in $U^{+}$. Let $ap^{+}$ be
the weight of edges inside $U^{+}$, $bp^{+}$ be the weight of
edges from $B$ to $U^{+}$, and $cp^{+}$ the weight of edges from
$U^{+}$ to $B$. Since all of the edges from $U^{+}$ to $B$ are not
misplaced, we have $c=1-2\gamma^{+}+\epsilon^{+}$, where
$0\leq\epsilon^{+}\leq\gamma^{+}$. We also have $2a+b+c=1$ and
$a+b=\gamma^{+}$, so $a+c=1-\gamma^{+}$,
$a=\gamma^{+}-\epsilon^{+}$ and $b=\epsilon^{+}$. Similarly, we
have $\left(1-2\gamma^{-}+\epsilon^{-}\right)p^{-}$ edges from
$U^{-}$ to $B$ and $\epsilon^{-}p^{-}$ edges from $B$ to $U^{+}$.
The expected weight of the edges in the cut (produced cut by
$f_{\delta}$) inside $B$ is \[
\frac{2-p^{-}-p^{+}-\left(1-2\gamma^{-}+2\epsilon^{-}\right)p^{-}-\left(1-2\gamma^{+}+2\epsilon^{+}\right)p^{+}}{8}=\]
\[
\frac{1-\left(1-\gamma^{-}+\epsilon^{-}\right)p^{-}-\left(1-\gamma^{+}+\epsilon^{+}\right)p^{+}}{4}\]

The expected weight of the edges in the cut from $U$ to $B$ is\[
\frac{\left(1-2\gamma^{-}+\epsilon^{-}\right)p^{-}+\left(1-2\gamma^{+}+\epsilon^{+}\right)p^{+}}{2}\]

The optimal cut must misplace at least $\delta$ of the endpoints
in $B$ and endpoints of weight $\gamma^{+}p^{+}+\gamma^{-}p^{-}$
in $U$, so the optimal cut has weight of at most\[
1-\frac{\delta(2-p^{-}-p^{+})+\gamma p^{+}+\gamma^{-}p^{-}}{2}\]

Setting $\epsilon^{-}=0=\epsilon^{+}$ lowers the expected weight
of the cut without affecting the weight of the bound for the
optimal cut, so we will assume this is the case. Define
$p=p^{+}+p^{-}$ and $\gamma$ such that $\gamma
p=\gamma^{+}p^{+}-\gamma^{-}p^{-}$. $2\gamma p$ is the weight of
misplaced endpoints in $U$ so $\gamma\leq\delta$. The
approximation ratio is no worse than\[
\frac{1+\left(1-3\gamma\right)p}{4\left(1-\delta\right)+2p\left(\delta-\gamma\right)}\]

Setting $\delta=\frac{1}{3}$ and differentiating the approximation
ratio shown according to $\gamma$ gives numerator (the denominator
is positive for $p\geq0$, $\gamma\leq\delta$) \[
-3p\left(\frac{8}{3}+2p\left(\frac{1}{3}-\gamma\right)\right)-2p\left(1+\left(1-3\gamma\right)p\right)=\]
\[
-p\left(10+4p-12p\gamma\right)\]

For $0\leq\gamma\le\delta$ and $p\geq0$, the derivative is non
positive, so maximizing $\gamma$ will give the lowest
approximation ratio. For $\gamma=\delta=\frac{1}{3}$, the
approximation ratio is at least  $\frac{3}{8}$.

A graph with two
vertices $X,Y$ with edge of weight $\frac{2}{3}$ from $X$ to $Y$
and an edge of weight $\frac{1}{3}$ from $Y$ to $X$ shows
$\frac{3}{8}$ to be an upper bound on the approximation ratio of
$f_{\frac{1}{3}}$.

We remark that a slightly larger value of $\delta$ can give an
approximation ratio better than 0.375, and in fact better than
0.39. This can be verified using the linear programming approach
of Theorem~\ref{linear program lower bound}.

\section{Finding approximation ratios via linear programs}
\label{sec:LP}

\begin{proof}
[Proof of Theorem \ref{linear program lower bound}] For a given
step function $f$, we present a linear program that constructs a
graph with the worst possible approximation ratio for the
oblivious algorithm that uses $f$ as a selection function.

Suppose that the set of discontinuity points of the step function $f$
is $0=z_{0}\leq z_{1}\leq z_{2}\leq\cdots\leq z_{n-1}\leq
z_{n}=1$. An {\em isolated} point (that is neither left continuous
nor right continuous) is counted as two discontinuity points, for
a reason that will become apparent later. In the graph produced by
the LP, a certain subset $S$ of vertices will correspond to the
optimal cut in the graph, $T_{i}$ corresponds to the set of
vertices in $S$ with bias between $z_{i-1}$ and $z_{i}$, and
$T_{i+n}$ corresponds to the set of vertices not in $S$ with bias between
$z_{i-1}$ and $z_{i}$ (A vertex with bias $z_{i}$ for
some $i$ can be chosen arbitrarily to be in one of the sets). We
assume that the weights of the edges are normalized such that the
weight of the cut corresponding to $S$ is~1. The variable $e_{ij}$
denotes the weight of the edges from the set $T_{i}$ to the set
$T_{j}$. Let $l,u:\left\{ 1..2n\right\} \to\left\{ 1..n\right\} $
be such that $T_{i}$ contains the set of vertices of biases
between $z_{l\left(i\right)}$ and $z_{u\left(i\right)}$
($l\left(i\right)<u\left(i\right)$).

We have the following constraints:

\begin{itemize}

\item$\underset{\substack{i\leq n\\
j>n}
}{\sum}e_{ij}=1$ - The weight of the cut is 1.

\item$\forall i\; z_{l\left(i\right)}\underset{j}{\sum}\left(e_{ij}+e_{ji}\right)\leq
\underset{j}{\sum}e_{ij}\leq
z_{u\left(i\right)}\underset{j}{\sum}\left(e_{ij}+e_{ji}\right)$ -
The (average) bias of the vertices is correct.

\item$\forall i,j\; e_{ij}\geq0$
- the weight of edges must be nonnegative.

\end{itemize}

Note that $e_{ii}$ appears twice in
$\underset{j}{\sum}\left(e_{ij}+e_{ji}\right)$, since it
contributes to both outweight and inweight.

Let $p_{i}=p_{i+n}=f\left(\frac{z_{i-1}+z_{i}}{2}\right)$ be the
probability of selecting a vertex in the sets $T_{i}$ and
$T_{i+n}$. (Here we used the convention that isolated points $z_i$
appear twice.)
The expected weight of the cut is
$\underset{i,j}{\sum}p_{i}\left(1-p_{j}\right)e_{ij}$, and this is
the approximation ratio of the oblivious algorithm on the graph if
the cut corresponding to $S$ is optimal.
Minimizing $\underset{i,j}{\sum}p_{i}\left(1-p_{j}\right)e_{ij}$
subject to the constraints gives a graph on which $f$ attains its
worst approximation ratio. There are some minor technicalities
involved in formally completing the proof, and the reader is
referred to Section~\ref{app:LP} in Appendix for these details.
\end{proof}


The linear program produces a graph that shows the worst approximation
ratio of the selection function used. The dual of the linear program
will produce a witness that the approximation ratio of the selection
function is not worse than the value output. Every constraint of the
dual has at most five variables, so it is technically possible to
check that the approximation ratio is correct. However, even in the
case of the function family from Section \ref{thm:There-Steps}, where
a function has three steps, the linear program has 15 constraints
and 36 variables and the dual has 15 variables and 36 constraints.
Hence, verifying the results manually becomes tedious. It is possible
to decrease the number of variables in the primal by using symmetrization
of the graph (adding another graph with inverted edges), but this
will complicate the description of the linear program and will increase
the probability of errors in the programming, while only reducing
the number of variables by half.

We implemented the linear program above. To gain confidence in our
implementation, we checked the linear program on the step function
form Theorem \ref{thm:There-Steps} and indeed got that the
approximation ratio is $\frac{3}{8}$. We also checked the program
on the uniformly random algorithm and got $\frac{1}{4}$
approximation ratio, as expected.

We now prove theorem \ref{0.483 approximation}. Define $f\left(x\right)$ to be 0 for $x<0.25$, 1 for $x>0.75$, for
any $0\leq i<100$, if $0.25+0.005i<x<0.25+0.005\left(i+1\right)$,
$f\left(x\right)=0.005+0.01i$,
$f\left(\frac{1}{2}\right)=\frac{1}{2}$, and right or left
continuous on all other points. Using the corresponding linear program,         as defined in theorem \ref{linear program lower bound}, we have determined that  the
approximation ratio of $f$ is more than 0.4835 but not more than
0.4836 (according to the values of the primal and the dual), as
claimed in Theorem \ref{0.483 approximation}. $f$ can be seen as a
discretized version of the function $g\left(x\right)=\max\left\{
0,\min\left\{ 1,2\left(x-\frac{1}{2}\right)+\frac{1}{2}\right\}
\right\} $, and we believe that the approximation ratio of $g$ is
slightly better. In principle, it is possible to show this, using
a finer discretized version of the function. However, it is too
time consuming to check this, so we did not do it.

The proof of Theorem~\ref{search best} is based on a standard
discretization argument. See Section~\ref{app:LP} in the appendix
for more details.

\section{An upper bound on oblivious approximation ratios}

To prove Theorem \ref{upper bound} we construct two weighted
graphs, $G_1$ and $G_2$. To get a good approximation ratio for
$G_{1}$, the probability of selecting a vertex with bias
$\frac{5}{9}$ needs to be close to 1/2, whereas for $G_{2}$ it
needs to be far from 1/2. Combining the two graphs gives a single
graph that upper bounds the approximation ratio of any oblivious
algorithm.
We remark that a linear program similar to that of
Section~\ref{sec:LP} assisted us in constructing $G_1$ and $G_2$.
\begin{example}
\label{greedyexample}$G_{1}$ is the following weighted graph:

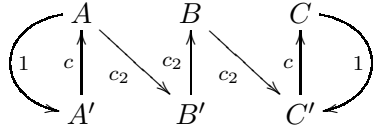
\[
\xymatrix{A\ar@(l,l)[d]^{1}\ar[dr]_{c_{2}} & B\ar[dr]_{c_{2}} & C\ar@(r,r)[d]_{1}\\
A'\ar[u]^{c} & B'\ar[u]^{c_{2}} & C'\ar[u]^{c}}
\]

where $c_{2}=c^{2}-1$.
\end{example}
Note that:
\begin{itemize}
\item The bias of $A$ and $A'$ is $\frac{c}{c+1}$.
\item The bias of $B$ and $B'$ is $\frac{1}{2}$.
\item The bias of $C$ and $C'$ is $\frac{1}{c+1}$.
\item There is a cut of weight $2c^{2}$ by selecting $A$, $B$, and $C$.
\end{itemize}
Let $\alpha$ be the probability of selecting a vertex with bias
$\frac{c}{c+1}$ for some oblivious algorithm (then the probability
of selecting a vertex with bias $\frac{1}{c+1}$ is $1-\alpha$).
Then the expected value of a solution produced by the algorithm is
\[
2\alpha\left(1-\alpha\right)\left(1+c\right)+\left(\alpha+\frac{1}{4}\right)\left(c^{2}-1\right)\]
And the approximation ratio is most\[
\frac{2\alpha\left(1-\alpha\right)\left(1+c\right)+\left(\alpha+\frac{1}{4}\right)\left(c^{2}-1\right)}{2c^{2}}\]

\begin{example}
$G_{2}$ is the following weighted graph:

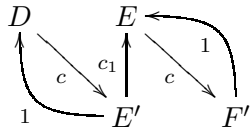
\[
\xymatrix{D\ar[dr]_{c} & E\ar[dr]_{c}\\
 & E'\ar[u]^{c_{1}}\ar@(l,d)[ul]^{1} & F'\ar@(u,r)[lu]^{1}}
\]

where $c_{1}=c-1$.
\end{example}
Note that:
\begin{itemize}
\item The bias of $D$ is $\frac{c}{c+1}$.
\item The bias of $E$ and $E'$ is $\frac{1}{2}$.
\item The bias of $F'$ is $\frac{1}{c+1}$.
\item There is a cut of weight $2c$ by selecting $D$ and $E$.
\end{itemize}
Let $\alpha$ be the probability of selecting the vertex $D$ (and
$1-\alpha$ is the probability of selecting the vertex $F'$).

The expected weight of the cut is\[
c\alpha+\frac{c-1}{4}+1-\alpha\] The approximation ratio is\[
\frac{1+\left(\alpha+\frac{1}{4}\right)\left(c-1\right)}{2c}\]

\ignore{Using Theorem 3.1 in \cite{W45} and the game defined in
Section \ref{sec:Mix-versus-Max} when there are finitely many
graphs and infinitely many algorithms, we know that there is a
distribution on the graphs that has approximation ratio that is
not better than the minimal approximation ratio. Then, we can
combine the two graphs into a single graph with the same
approximation ratio.}

Consider a graph composed of one copy of $G_1$ and three copies of
$G_2$. The approximation ratio is at most \[
\frac{2\alpha\left(1-\alpha\right)\left(1+c\right)+\left(\alpha+\frac{1}{4}\right)\left(c^{2}-1\right)+3+3\left(\alpha+\frac{1}{4}\right)\left(c-1\right)}{2c^{2}+6c}\]
which, for fixed $c$, is a parabola with a maximal point.

For $c=1.25$, the approximation ratio is \[
\frac{213+372\alpha-288\alpha^{2}}{680}\]
the maximum is achieved at $\alpha=\frac{31}{48}$, and the value
there is $\frac{533}{1088}<0.4899$. Hence, no algorithm based on
oblivious algorithms (maximum of several oblivious algorithms or choosing
one to use according to some distribution) can achieve better approximation
ratio and this graph proves Theorem \ref{upper bound}.

The proof of Theorem \ref{nonsymmetric} appears in
Section~\ref{app:nonsymmetric} in the appendix.

\bibliographystyle{plain}
\bibliography{t}

\begin{appendix}

\section{Some remarks on 2-AND versus DICUT}
\label{sec:andcutremarks}

\ Despite the fact that the reduction from 2-AND to DICUT does not preserve  the weight of an optimal solution (a cut may have larger weight than the weight of any assignment), it
is possible to use a more generalized version of an oblivious
algorithm to generate only such cuts that can be transformed to
valid assignments. Instead of selecting $x$ and $\bar{x}$ to be in
$S$ independently, choose $x$ to be in $S$ according to the
selection function and set $\bar{x}$ to be in $S$ iff $x\notin S$.
The probability of $x$ and $\bar{x}$ to be in $S$ is the same as
before, and since there are no edges between then, the approximation ratio is the same,  due to the
linearity of expectation.

This can be generalized further: Instead of choosing vertices independently,  Divide any graph into disjoint
independent sets. Fix a selection function $f$. In each set, the marginal probability of a vertex being in $S$, will be the same as dictated by $f$ (however, the choices inside each independent set need not be independent). Then, due to the linearity of expectation, the
approximation ratio is the same.

For example, if the graph is $G_{\varphi}$ , and the selection
function decides that all vertices will be in $S$ with probability
$\frac{1}{2}$, by choosing both $x$ and $\bar{x}$ to be in $S$ or both not in $S$, we can find a cut with
the property that $x\in S$ iff $\bar{x}\in S$, with the same
expected weight of a valid assignment.

\section{Notes on linear programs}
\label{app:LP}

Here we provide some details required for a formal proof of
Theorem~\ref{linear program lower bound}.

Suppose that $r$ is the minimum value of the linear program and
that $\forall i\; e_{ii}=0$ for the optimal solution. Define a
vertex for each $T_{i}$ and an edge $\left(T_{i},T_{j}\right)$
with weight $e_{ij}$ for all $i,j$. This is a graph with the
property that $f$ achieves an approximation ratio $r$ if no vertex
has bias $z_{i}$ for some $i$. However, this is a minor obstacle;
add a vertex with total weight $\epsilon>0$ arbitrarily small,
with edges to or from all vertices with biases exactly $z_{i}$, so
that their biases will change slightly and the probability of
selecting the vertex $T_{i}$ will be $p_{i}$. The infimum of the
approximation ratios on the graphs (as $\epsilon\to0$) will be
$r$. Now, assume that for some $i$'s $e_{ii}>0$. Construct the
previous graph (without self loops). For $i$ such that $e_{ii}>0$,
split the vertex $T_{i}$ into two vertices, $A_{i}$ and $B_{i}$.
Every edge with an endpoint of $T_{i}$ will be split to two edges,
each with half the weight such that one will have endpoint $A_{i}$
instead of $T_{i}$, and the other will have endpoint $B_{i}$
instead of $T_{i}$. Add the edges $\left(A_{i},B_{i}\right)$ and
$\left(B_{i},A_{i}\right)$, each with weight $\frac{e_{ii}}{2}$.
All the constraints hold for the graph.

\begin{proof}
[Proof Sketch of Theorem \ref{search best}] Consider
$\mathcal{F}$, the family of $n^{n+1}$ antisymmetric step
functions that are constant on each of the $2n$ intervals of width
$\frac{1}{2n}$ of the unit interval, and the value on each of
those intervals is of the form $\frac{k}{n}$ with
$k\in\mathbb{N}$, and are left or right continuous between the
intervals. Also, in order to be antisymmetric, $\forall
f\in\mathcal{F}\: f\left(\frac{1}{2}\right)=\frac{1}{2}$. As a
corollary from the proof of the linear program, left or right
continuity of a step function does not change the approximation
ratio, so there are indeed only $n^{n+1}$ functions (due to
antisymmetry) to consider when looking at the approximation ratio.
Using $n^{n+1}$ linear programs (time
$\mathrm{poly}\left(n\right)n^{n}$) it is possible to find the
function with the best approximation ratio from the set.

It is possible that the best function from the set is not the best
possible selection function. However, it is close. Suppose that the best selection function is a step function that is
constant on the same intervals, but may have any value on those
intervals. Let $g$ be such a step function, and let $f$ be the
closest function (in $\ell_{\infty}$ distance) from $\mathcal{F}$.
Then, the probability an edge being in the cut when using $g$
instead of $f$ is at most $\mathrm{O}\left(\frac{1}{n}\right)$
larger, so the approximation ratio of $f$ is at most
$\mathrm{O}\left(\frac{1}{n}\right)$ lower than $g$.

Now, fix any selection function $h$. Let $g$ be a step function
that is constant on each of the $2n$ intervals of width
$\frac{1}{2n}$ of the unit interval such that for all
$k\in\mathbb{N}$ with $k\leq2n$,
$g\left(\frac{k}{2n}\right)=h\left(\frac{k}{2n}\right)$. Given any
{}``bad'' graph for $g$, by adding new edges and a single vertex,
it can be transformed into a {}``bad'' graph for $h$, and the
approximation ratio will be at most
$\mathrm{O}\left(\frac{1}{n}\right)$ lower. Thus, the
approximation ratio of $g$ is at most
$\mathrm{O}\left(\frac{1}{n}\right)$ lower than $h$.

Therefore, the approximation ratio of the best function from
$\mathcal{F}$ has approximation ratio worse by at most
$\mathrm{O}\left(\frac{1}{n}\right)$ than any oblivious algorithm.
\end{proof}

\section{Weighted versus unweighted graphs}
\label{app:weighted}

\begin{lem}
Any weighted graph $G$ with rational weights can be transformed to
an unweighted graph $G'$ such that for any oblivious algorithm the
approximation on $G'$ will not be better than the approximation
ratio on $G$.\end{lem}
\begin{proof}
Let $G=\left(V,E,w\right)$. Define $W$ to be $\max_{e\in
E}w\left(e\right)$. Define $w'$ to be $\frac{w}{W}$.
$\left(V,E,w'\right)$ is a weighted graph with rational weights,
the maximal weight is 1, and the approximation ratio is the same
as the approximation ratio for $G$. There are
$w_{e},M\in\mathbb{N}$ such that
$w'\left(e\right)=\frac{w_{e}}{M}$ for all $e\in E$.

Let $V'$ be composed of $M$ copies of $V$. $v\in V$ will be
identified with $\left\{ v_{1},\cdots,v_{M}\right\} \subseteq V$.
For every $e=\left(v,u\right)\in E$, create a $w_{e}$-regular
bipartite graph between $\left\{ v_{1},\cdots,v_{M}\right\} $ and
$\left\{ u_{1},\cdots,u_{M}\right\} $ in $E'$ (directed towards
$\left\{ u_{1},\cdots,u_{M}\right\} $).

$G'=\left(V',E'\right)$ satisfies the condition of the lemma.
\end{proof}

We can now prove Proposition \ref{weighted=00003Dunweighted}.
Given an arbitrary selection function and arbitrary $\epsilon >
0$, the proof of Theorem~\ref{search best} shows that the
selection function can be replaced by a step function that is
continuous on all irrational values, and the approximation ratio
deteriorates by at most $\epsilon$. Any weighted graph $G$ has
finitely many vertices and biases, and since the selection
function is continuous at all irrational values, $G$ can be
transformed to $G'$ with only rational weights and an
approximation ratio higher by at most $\epsilon$ (by changing each
irrational weight by a small value). Letting $\epsilon$ tend to~0
proves Proposition \ref{weighted=00003Dunweighted}.


\section{Selection functions that are not antisymmetric}
\label{app:nonsymmetric}

\begin{proof}[Proof of Theorem \ref{nonsymmetric}]

Let $G=\left(V,E,w\right)$ and create $G'=\left(V,E',w'\right)$
from $G$ by inverting all edges. That is, $E'=\left\{
\left(v,u\right)|\left(u,v\right)\in E\right\} $ and
$w'\left(\left(v,u\right)\right)=w\left(\left(u,v\right)\right)$.
Let $G''$ be the disjoint union of these graphs. Consider a
selection function $f$ that is not antisymmetric. Let
$g\left(x\right)=\frac{f\left(x\right)+1-f\left(1-x\right)}{2}$.
$g$ is antisymmetric. Let $\left(u,v\right)$ be an edge, where the
bias of $u$ is $r_{u}$ and the bias of $v$ is $r_{v}$. The
probability of selecting $\left(u,v\right)$ when using $g$ is
$g\left(r_{u}\right)\left(1-g\left(r_{v}\right)\right)$ or
$\frac{1}{4}\left(f\left(r_{u}\right)+1-f\left(1-r_{u}\right)\right)\left(f\left(1-r_{v}\right)+1-f\left(r_{v}\right)\right)$.
The expected weight contributed by $\left(u,v\right)$ and
$\left(u,v\right)$ when using $g$ is \[
\frac{1}{2}\left(f\left(r_{u}\right)+1-f\left(1-r_{u}\right)\right)\left(f\left(1-r_{v}\right)+1-f\left(r_{v}\right)\right)\]
 and when using $f$ the expected weight is $f\left(r_{u}\right)\left(1-f\left(r_{v}\right)\right)+f\left(1-r_{v}\right)\left(1-f\left(1-r_{u}\right)\right)$.
The advantage of using $g$ over $f$ is \[
\frac{1}{2}\left(1-f\left(r_{u}\right)-f\left(1-r_{u}\right)\right)\left(1-f\left(r_{v}\right)-f\left(1-r_{v}\right)\right)\]
which is positive if $\forall z\:
f\left(z\right)+f\left(1-z\right)\geq1$ or $\forall z\:
f\left(z\right)+f\left(1-z\right)\le1$.

Recall that the proof of Theorem \ref{upper bound} is based on a
graph whose vertices have biases $\frac{1}{2}$, $\frac{c}{c+1}$,
and $\frac{1}{c+1}$. Hence if
$f\left(\frac{1}{2}\right)=\frac{1}{2}$, the upper bound holds for
$f$, regardless of the antisymmetry of $f$.

If $f\left(\frac{1}{2}\right)=\frac{1}{2}+\delta$, since
$\left|1-f\left(x\right)-f\left(1-x\right)\right|\leq1$, the
approximation ratio can increase by at most $\delta$ times the
weight of all edges (compared to using the antisymmetric version
of the function). However, the approximation ratio for an even
cycle will be $\frac{1}{2}-2\delta^{2}$. Therefore, there is
$\gamma>0$ such that no approximation better than
$\frac{1}{2}-\gamma$ can be achieved to Max DICUT using oblivious
algorithms, even if the selection function is not
antisymmetric.\end{proof}

\end{appendix}

\end{document}